\documentclass[runningheads]{llncs}
\usepackage{graphicx}
\usepackage[misc]{ifsym} 
\usepackage[colorlinks,allcolors=blue]{hyperref}
\usepackage{subcaption} 
\usepackage{pgfplots}
\pgfplotsset{compat=newest}

\usepackage{mathtools}
\usepackage{paralist} 
\usepackage{amsthm}
\usepackage{amsmath}
\usepackage{amssymb}
\usepackage{relsize}
\usepackage{tikz} 
\usepackage{varwidth}
\usetikzlibrary{arrows.meta,arrows}
\usetikzlibrary{shapes, positioning}

\usepackage[nameinlink,noabbrev]{cleveref} 
\Crefformat{figure}{#2Fig.~#1#3}
\Crefformat{equation}{#2Eq.~(#1)#3}
\Crefmultiformat{figure}{Figs.~#2#1#3}{ and~#2#1#3}{, #2#1#3}{ and~#2#1#3}

\usepackage{nameref}

\theoremstyle{plain}
\newtheorem{thm}{Theorem} 
\theoremstyle{definition}
\newtheorem{defn}{Definition} 
\usepackage{url}
\usepackage{multirow}
\urldef{\mailsa}\path|author1@someinst.something|
\urldef{\mailsb}\path|author2@someotherinst.something|

\usepackage{slashbox}

\usepackage{multicol}
\usepackage{algorithm} 
\usepackage[ruled, vlined, linesnumbered,algo2e]{algorithm2e}
\usepackage[noend]{algpseudocode}
\usepackage{etoolbox}

\makeatletter
\newcommand*{\algrule}[1][\algorithmicindent]{%
  \makebox[#1][l]{%
    \hspace*{.2em}
    \vrule height .75\baselineskip depth .25\baselineskip
  }
}

\newcount\ALG@printindent@tempcnta
\def\ALG@printindent{%
    \ifnum \theALG@nested>0
    \ifx\ALG@text\ALG@x@notext
    \else
    \unskip
    \ALG@printindent@tempcnta=1
    \loop
    \algrule[\csname ALG@ind@\the\ALG@printindent@tempcnta\endcsname]%
    \advance \ALG@printindent@tempcnta 1
    \ifnum \ALG@printindent@tempcnta<\numexpr\theALG@nested+1\relax
    \repeat
    \fi
    \fi
}
\patchcmd{\ALG@doentity}{\noindent\hskip\ALG@tlm}{\ALG@printindent}{}{\errmessage{failed to patch}}
\patchcmd{\ALG@doentity}{\item[]\nointerlineskip}{}{}{} 
\makeatother

\newcommand{\RNum}[1]{\uppercase\expandafter{\romannumeral #1\relax}} 

\usepackage{listings}
\usepackage[utf8]{inputenc}
\graphicspath{ {figures/} }
\usepackage{array}
\usepackage[nottoc,notlot,notlof]{tocbibind}

\usepackage{chngcntr}

\let\emptyset\varnothing

\usepackage{etoolbox}
\patchcmd{\paragraph}{\itshape}{\bfseries\boldmath}{}{}

    \newenvironment{customlegend}[1][]{%
        \begingroup
        \csname pgfplots@init@cleared@structures\endcsname
        \pgfplotsset{#1}%
    }{%
        \csname pgfplots@createlegend\endcsname
        \endgroup
    }%

    \def\addlegendimage{\csname pgfplots@addlegendimage\endcsname}
\usetikzlibrary{patterns}
\usepackage{enumitem} \setlist[enumerate]{label={\upshape (\roman*)}}
\newcounter{centredequ} 
\newenvironment{centredequ}{\refstepcounter{equation}\hfill\begin{math}}{\end{math}\hfill$(\theequation)$\par\noindent}
\crefname{centredequ}{eq.}{eqs. }
\Crefname{centredequ}{Eq.}{Eqs. }

\begin{document}
\title{A Projected Upper Bound for Mining High Utility Patterns from Interval-Based Event Sequences}
\author{S. Mohammad Mirbagheri\textsuperscript{(\Letter)}}
%
%
%
\institute{Department of Computer Science, University of Regina, Regina, Canada\\
\email {Mo.Mirbagheri}@uregina.ca}

\maketitle              
\begin{abstract}

High utility pattern mining is an interesting yet challenging problem. The intrinsic computational cost of the problem will impose further challenges if efficiency in addition to the efficacy of a solution is sought.
Recently, this problem was studied on interval-based event sequences with a constraint on the length and size of the patterns. However, the proposed solution lacks adequate efficiency. To address this issue, we propose a projected upper bound on the utility of the patterns discovered from sequences of interval-based events.  
To show its effectiveness, the upper bound is utilized by a pruning strategy employed by the HUIPMiner algorithm. Experimental results show that the new upper bound improves HUIPMiner performance in terms of both execution time and memory usage. 

\end{abstract}
\section{Introduction}
\textit{Frequent Pattern Mining} (FPM) \cite{FPMSurvey} has been well-studied over the past two decades. The goal of FPM is to discover patterns such that the frequency of their appearances in a dataset is higher than a user-specified threshold. Since the measure of the interestingness of the patterns is frequency, FPM is incapable of addressing problems where the frequency of occurrences of patterns is not of interest. As a result, High Utility Pattern Mining (HUPM) was proposed for problems where patterns with high utilities, e.g., profits generated by patterns, are of interest, and thus they are measured based on their utilities rather than the frequency of occurrences. In particular, discovering patterns with utilities no less than a minimum utility threshold set by a user is the focus of HUPM. 
 
Depending on the domain of the data, and similar to FPM, various types of HUPM have been introduced, e.g., High Utility Itemest Mining \cite{UtilityItemsetSurvey}, High Utility Episode Mining \cite{wu2013}, and High Utility Sequential Pattern Mining \cite{UtilitySequenceSurvey}. 
Interval-based event sequences (e-sequences) are the sequences in which multiple events can occur coincidentally and persist over varying periods of time. 
E-sequences are present in many real-world applications from different domains, such as medicine \cite{patelHepatit,2019medical}, sensor technology \cite{morchenSensor}, sign language \cite{signLanguage}, and activity recognition \cite{motion2016}.
As a running example, four sequences of interval-based events which form a dataset are presented in \Cref{DB}. 
As shown, each e-sequences contains event intervals with various labels, beginning and finishing times. These e-sequences have been visualized in the furthest right column of the table. 

Mirbagheri and Hamilton \cite{Mirbagheri-DKE,Coincidence} have recently studied HUPM on e-sequences. They introduced a framework to incorporate the concept of utility into these sequences and proposed an algorithm named HUIPMiner to discover the high utility patterns. They defined an upper bound on the utility of e-sequences w.r.t. a maximum length $k$, namely the L-\textit{sequence-weighted utilization} ($\mathrm{LWU}_k$). The upper bound is employed in a pruning strategy in the algorithm to reduce the search space, which results in reducing the execution time and space. In this paper, we improve their work by deriving a new upper bound on the utility of e-sequences, namely the \textit{Projected utilization} ($\mathcal{P}_{k}$), and show both theoretically and empirically that if $\mathcal{P}_{k}$ is employed, the execution of the algorithm will improve compared with when the algorithm utilizes $\mathrm{LWU}_k$.    

The remainder of this paper is as follows. \Cref{Section2} reviews the background related to the sequences of interval-based events and the preliminaries relevant to the high utility pattern mining. \Cref{Sec.Down} reviews and also introduce properties that are used in reducing the search space. \Cref{Section3} introduces the projected upper bound and its consequent property.  
\Cref{Section4} reports on the empirical results and evaluates the proposed solution.  
\Cref{Section5} concludes the paper. 


\renewcommand{\arraystretch}{1.15} 
\begin{table}[H] 
\centering
\caption{Example of an E-sequence dataset}
\label{DB}
\begin{tabular}{|c|p{1.3cm}|p{1.9cm}|p{1.59cm}|l|}
\hline
\textbf{ID} & \textbf{Event Label} &  \textbf{Beginning Time} & \textbf{Finishing Time} &   \multicolumn{1}{>{\centering\arraybackslash}m{6cm}|}{\textbf{Pictorial Example}}  \\ \hline
\multirow{4}{*}{1} & $A$ & 6& 12& \multirow{4}{*}{  \includegraphics[width=0.5\textwidth, height=0.08\textheight]{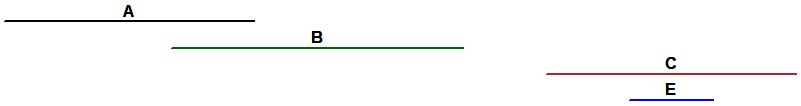}} \\ 
        & $B$ & 10& 17& \\ 
        & $C$ & 19& 25& \\ 
        & $E$ & 21& 23& \\ \hline
\multirow{5}{*}{2} & $A$ & 2& 7& \multirow{5}{*}{ \includegraphics[width=0.5\textwidth, height=0.1\textheight]{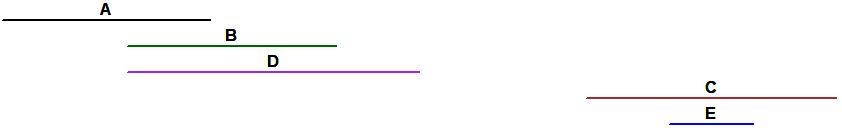}  } \\ 
        & $B$ & 5& 10& \\
        & $D$ & 5& 12& \\ 
        & $C$ & 16& 22& \\ 
        & $E$ & 18& 20& \\ \hline
\multirow{4}{*}{3} & $B$ & 6& 12& \multirow{4}{*}{  \includegraphics[width=0.5\textwidth, height=0.08\textheight]{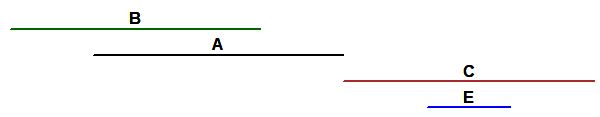}} \\ 
        & $A$ & 8& 14& \\ 
        & $C$ & 14& 20& \\ 
        & $E$ & 16& 18& \\ \hline
\multirow{4}{*}{4} & $B$ & 1& 5& \multirow{4}{*}{  \includegraphics[width=0.5\textwidth, height=0.08\textheight]{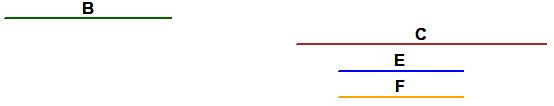}} \\ 
        & $C$ & 8& 14& \\
        & $E$ & 9& 12& \\ 
        & $F$ & 9& 12& \\ \hline

	\end{tabular}
\end{table}
\section{Background}
\label{Section2}
In this section, we review the preliminaries of the HUIPM problem \cite{Mirbagheri-DKE,Coincidence} that will be used to derive the projected upper bound.  

Let $\Sigma=\{A,B,...\}$ denote a finite alphabet. A triple $e=(l,b,f)$, where $l \in \Sigma$ is the event label, $b \in \mathbb{N} $ is the beginning time, and $f \in \mathbb{N}$ is the finishing time $(b<f)$, is said to be an \textit{event-interval}.
A list $s=\langle e_1,e_2, ...,e_n\rangle $ containing $n$ event intervals, which are ordered based on beginning time in ascending order while ties are broken based on the lexicographical order of the event labels, is called an event-interval sequence or \textit{E-sequence}. 
The number of event-intervals in $s$ determines the \textit{size} of E-sequence $s$ (denoted as $|s|=n$).
A set $D=~\{s_1, s_2, ..., s_d \}$ containing $d$ E-sequences, where each E-sequence $s_i$ is associated with an unique identifier $1 \leq i \leq d$, is called an \textit{E-sequence dataset}. For example, \Cref{DB} depicts an E-sequence dataset consisting of four E-sequences with identifiers 1 to 4.
\subsection{Coincidence}
\begin{defn} 
\textit{E-sequence unique time points} $T_s =\langle t_1, t_2, ..., t_m \rangle$ is a finite non-empty sequence consisting of the unique time points of $s$ sorted in ascending order such that $t_k < t_{k+1}$, $1 \leq k \leq m-1$, $t_k \in \{ b \ \lor f  \ | \ b, f \in s \}$.
\end{defn}
\begin{defn} 
\label{coincidenceDEF}
Let $s=\langle (l_1, b_1, f_1), ...,(l_j, b_j, f_j), ...,(l_n, b_n, f_n)\rangle$ be an E-sequence. A function $\Phi_s: \mathbb{N} \times \mathbb{N} \rightarrow 2^{\Sigma} $ is defined as:
\begin{equation}
\Phi_s(t_q,t_{q'})=  \{ l_j \ | \ (l_j, b_j, f_j)  \in s \ \wedge \ ( b_j \leq t_q) \wedge (t_{q'} \leq f_j)   \}
\end{equation}
where $1\leq j \leq n$ and $t_q< t_{q'}$.
 Given an E-sequence $s$ with corresponding E-sequence unique time points $T_s =\langle t_1, t_2, ..., t_m \rangle$, a \textit{coincidence} $c_k$ is defined as $\Phi_s(t_k,t_{k+1})$ where $t_k, t_{k+1} \in T_s$, $1 \leq k \leq m-1 $, are two consecutive time points. The \textit{duration} $\lambda_k$ of coincidence $c_k$ is $t_{k+1}-t_k$. 
The \textit{size of a coincidence} is the number of event labels in the coincidence. 
\end{defn}
For example, the E-sequence unique time points of $s_4$ in \Cref{DB} is $T_{s_4}= \{1,5,8,9,12,14\}$. Coincidence $c_4=\Phi_{s_4}(9,12)= \{C,E,F\}$, $\lambda_4=12-9=3$, and $|c_4|=3$.
\begin{defn} 
 A coincidence label sequence, or \textit{L-sequence}, $L= \langle c_1c_2...c_{g} \rangle$ is an ordered list of $g$ coincidences.
The length of an L-sequence is denoted by $K$, iff there are exactly $K$ coincidences in the L-sequence.
The size of an L-sequence, denoted $Z$, is determined by the maximum size of the coincidences in the L-sequence.
\end{defn}
For example, since $L= \langle \{C\} \{A,B\} \{D\} \rangle $ has 3 coincidences and the maximum size of the coincidences is $max\{1,2,1\}=2$, the length and size of $L$ is 3 and 2, respectively.

\begin{defn} 
Given a coincidence $c_k$ in E-sequence $s$, a coincidence eventset, or \textit{C-eventset}, is denoted $\sigma_k$ and defined as an ordered pair consisting of the coincidence $c_k$ and the corresponding coincidence duration $\lambda_k$, i.e.:
 \begin{equation}
\sigma_k=(c_k, \lambda_k)   
\end{equation} 
For brevity, the braces are omitted if $c_k$ in C-eventset $\sigma_k$ has only one event label, which is referred to as a \textit{C-event}.
 A coincidence eventset sequence, or \textit{C-sequence}, is an ordered list of C-eventsets, which is defined as $C = \langle \sigma_1 \sigma_2 ... \sigma_{h} \rangle$, where $h=|T_s|-1$. 
A \textit{C-sequence dataset} $\delta$ consists of a set of C-sequences, where each $C$ is associated with a unique identifier. 
\end{defn} 
For example, a C-sequence dataset, which includes C-sequences corresponding to the E-sequences shown in \Cref{DB}, is presented in \Cref{CoincidenceDB}. We denote the C-sequence with identifier 1 as $C_{\mathrm{s}_1}$; other C-sequences are numbered accordingly.
One can notice that in addition to describing E-sequences in a formulated language by transforming them to C-sequences, this representation also captures the durations of the event intervals.
\begin{table}[H]
\centering
\caption{ C-sequence dataset corresponding to the E-sequences in \Cref{DB}}
\label{CoincidenceDB}
\begin{tabular}{|c|c|}
\hline
ID & C-sequence   \\ \hline
 1 &  $\langle(A,4)(\{A,B\},2)(B,5)(\emptyset,2)(C,2)(\{C,E\},2) (C,2) \rangle$        \\ \hline
 2 &  $\langle (A,3)(\{A,B,D\},2)(\{B,D\},3)(D,2) (\emptyset,4)(C,2)(\{C,E\},2) (C,2) \rangle$         \\ \hline
 3 &  $\langle(B,2) (\{A,B\},4)(A,2)(C,2)(\{C,E\},2)(C,2) \rangle$         \\ \hline
 4 &$\langle(B,4)(\emptyset,3)(C,1)(\{C,E,F\},3)(C,2) \rangle$\\ \hline
\end{tabular}
\end{table}
\begin{defn}
Given two C-eventsets $\sigma_a=( c_a , \lambda_a)$
and  $\sigma_b=(c_b , \lambda_b)$, $\sigma_b$ \textit{contains} $\sigma_a$, which is denoted $\sigma_a \subseteq \sigma_b$, iff $ c_a \subseteq c_b \wedge \ \lambda_a=\lambda_b$.
Given two C-sequences $C=\langle \sigma_1 \sigma_2 ... \sigma_{h} \rangle$ and $C'=\langle \sigma_1' \sigma_2' ... \sigma_{h'}' \rangle$, we say $C$ is a \textit{C-subsequence} of $C'$, denoted $C \subseteq C'$, iff there exist integers $1 \leq j_1 \leq j_2 \leq ... \leq j_h \leq h' $ such that $\sigma{_k} \subseteq \sigma_{j_k}'$ for $1 \leq k \leq h$.
Given a C-sequence $C=\langle \sigma_1 \sigma_2 ... \sigma_{h} \rangle= \langle (c_1, \lambda_1)(c_2,\lambda_2)...(c_h, \lambda_h) \rangle$ and an L-sequence $L= \langle c_1'c_2'...c_{g}' \rangle$, $C$ \textit{matches} $L$, denoted as $C \sim L$, iff $h = g$ and $c_k=c_k'$ for $1 \leq k \leq h$.
\end{defn}
\sloppy For example, $\langle (A,4)\rangle$, $\langle (\{A,B\},2)(B,5) \rangle$, and $\langle (\{A,B\},2) \rangle$, are C-subsequences of C-sequence $C_{\mathrm{s}_1}$, while $ \langle (\{A,B,D\},2) \rangle$ and $ \langle (\{A,B\},2)(B,2) \rangle$ are not.
It is possible that multiple C-subsequences of a C-sequence match a given L-sequence. For example, if we want to find all C-subsequences of $C_{\mathrm{s}_1}$ in \Cref{CoincidenceDB} that match the L-sequence $\langle B \rangle$, we obtain $\langle (B,2) \rangle$ in the second C-eventset and $\langle (B,5) \rangle$ in the third C-eventset.

\subsection{Utility}
\label{UtilSec}
Let each event label $l \in \Sigma$, be associated with a value, called the \textit{external utility}, which is denoted as $p(l)$, such that $p: \Sigma \rightarrow \mathbb{R}_{\geq 0}  $. The external utility of an event label may correspond to any value of interest, such as the unit profit that is associated with the event label. In the following examples, we use the values presented in \Cref{External} as the external utilities associated with the C-sequence dataset shown in \Cref{CoincidenceDB}.
\begin{table} 
\centering
\caption{ External utilities associated with the event labels }
\label{External}
\begin{tabular}{|c|c|c|c|c|c|c|c|}
\hline
 Event label & A &B&C&D&E&F& $\emptyset$         \\ \hline
 External utility & 2&1&1&3&2&5 &0       \\ \hline
\end{tabular}
\end{table}

Let the utility of a C-event $(l,\lambda)$ be $\mathrm{u}(l,\lambda)= p(l) \times \lambda$.
The utility of a C-eventset $\sigma=( c , \lambda)= (\{l_1,l_2, ..., l_{|c|}\},\lambda)$ is defined as:
$
\mathrm{u_e}(\sigma)= \sum_{i=1}^{|c|} \mathrm{u}(l_i,\lambda)
$. The utility of a C-sequence $C=\langle \sigma_1 \sigma_2 ... \sigma_{h} \rangle$ is defined as:
$\mathrm{u_s}(C)= \sum_{i=1}^{h} \mathrm{u_e}(\sigma_i)$.
Therefore, the utility of the C-sequence dataset $\delta =\{ C_{s_1}, C_{s_2}, ..., C_{s_r} \} $ is defined as:
$
\mathrm{u_d}(\delta)= \sum_{i=1}^{r} \mathrm{u_s}(C_{s_{i}})
$.
For example, the utility of C-sequence $C_{s_1}= \langle(A,4)(\{A,B\},2)(B,5)(\emptyset,2)(C,2)(\{C,E\},2) (C,2) \rangle$ is $\mathrm{u_s}(C_{s_1})= 4 \times 2+ 2 \times (2+1) + 5 \times 1+ 2 \times 0 +2 \times 1+ 2 \times (1+ 2)+2 \times 1 = 29 $, 
and the utility of the C-sequence dataset $\delta$ in \Cref{CoincidenceDB} is $\mathrm{u_d}(\delta)=\mathrm{u_s}(C_{s_1})+\mathrm{u_s}(C_{s_2})+\mathrm{u_s}(C_{s_3})+\mathrm{u_s}(C_{s_4})=29+46+28+31=134$.

\begin{defn}
\label{MaxUtilK}
The \textit{maximum utility of $k$ C-eventsets in a C-sequence} is defined as:
$
\mathrm{u_{max_k}}(C,k)=  max\{ \mathrm{u_s}(C ') \ | \ C' \subseteq C \ \wedge \ |C'| \leq k \ \}
$. 
\end{defn}
For example, the maximum utility of 2 C-eventsets in $C_{s_1}$ is $\mathrm{u_{max_k}}(C_{s_1},2)=max\{ \mathrm{u_s}(\langle (A,4)(\{A,B\},2) \rangle), \mathrm{u_s}(\langle (A,4)(\{C,E\},2) \rangle) \}=14$.
\begin{defn}
Given a C-sequence dataset $\delta$ and an L-sequence $L= \langle c_1c_2...c_g \rangle$, the utility of $L$ in C-sequence $C=\langle \sigma_1 \sigma_2 ... \sigma_{h} \rangle \in \delta$ is defined as a \textit{utility set}:
\begin{equation}
\mathrm{u_l}(L,C)= \bigcup_{C' \sim L \wedge C' \subseteq C} \mathrm{u_s}(C')
\end{equation}
Consequently, the utility of $L$ in $\delta$ is defined as:
\begin{equation}
\mathrm{u_l}(L)= \bigcup_{C \in \delta} \mathrm{u_l}(L,C)
\end{equation}
\end{defn}
For example, consider L-sequence $L=\langle \{A\} \{B\} \rangle$. The utility of $L$ in $C_{s_1}$ shown in \Cref{CoincidenceDB} is $\mathrm{u_l}(L,C_{s_1})=\{\mathrm{u_s}(\langle (A,4)(B,2) \rangle),\mathrm{u_s}(\langle (A,4)(B,5)\rangle), \mathrm{u_s}(\langle (A,2)(B,5) \rangle) \}=\{10,13,9\}$. Also, the utility of $L$ in $\delta$ is $\mathrm{u_l}(L)=\{\mathrm{u_l}(L,C_{s_1}), \mathrm{u_l}(L,C_{s_2})\}=\{\{10,13,9\},\{8,9,7\}\}$. 
As seen from the above example and also in contrast to a sequence in frequent sequential pattern mining, multiple utility values can be associated with an L-sequence. The possibility of having multiple utility values will lead us to the concept of high utility, which is explored briefly in the next section.    

\subsection{High Utility Interval-based Pattern Mining}
\begin{defn}
The \textit{maximum utility} of an L-sequence $L$ in C-sequence dataset $\delta$ is defined as $\mathrm{u_{max}}(L)$:
\begin{equation}
\label{MaxUtilEq}
\mathrm{u_{max}}(L)= \sum_{C \in \delta} \mathrm{max} (\mathrm{u_{l}}(L,C))
\end{equation}
\end{defn} 
For example, the maximum utility of an L-sequence $L=\langle \{A\}\{B\} \rangle$ in C-sequence dataset $\delta$ shown in \Cref{CoincidenceDB} is $\mathrm{u_{max}}(L)=13+9+0+0=22 $.

\begin{defn}
An L-sequence $L$ is a \textit{high utility interval-based pattern} iff its maximum utility is no less than a user-specified minimum utility threshold $\xi$. Formally:
$
\mathrm{u_{max}}(L) \geq \xi \iff L \text{ is a high utility interval-based pattern.}
$
\end{defn}
\paragraph{Problem \RNum{1}:}
Given a user-specified minimum utility threshold $\xi$, an E-sequence dataset $D$, and  external utilities for event labels, the problem of high utility interval-based mining is to discover all L-sequences such that their utilities are at least $\xi$. 
When the maximum length and size of the L-sequence are specified, one can make Problem \textbf{\RNum{1}} more specialized to give Problem \textbf{\RNum{2}}, which is to discover all L-sequences with lengths and sizes of at most $K$ and $Z$, respectively, such that their utilities are at least $\xi$.

\section{Downward Closure Property}
\label{Sec.Down} 
A tight upper bound on the utility of the candidates reduces the search space leading to a more efficient way of pattern discovery. Here, we review an upper bound on the utility of L-sequences, namely \textit{$\mathrm{LWU}_k$}, which leads to the L-sequence-weighted Downward Closure (LDC) property (\Cref{I}). This property can be utilized by an algorithm, i.e., it was previously employed by the HUIPMiner algorithm \cite{Mirbagheri-DKE,Coincidence}, to prune redundant candidates. We also introduce interesting properties, which are used later to construct and verify the projected upper bound. 

\begin{defn}(\textit{$\mathrm{LWU}_k$}) 
\label{Upper}
The L-sequence-weighted utilization of an L-sequence w.r.t. a maximum length $k$ is defined as:
\begin{equation}
\label{LWU_Eq}
\mathrm{LWU}_k(L) = \sum_{C' \sim L \wedge C' \subseteq C \wedge C  \in \delta } \mathrm{u_{max_k}}(C,k) 
\end{equation} 
\end{defn}
For example, the L-sequence-weighted utilization of $L=\langle \{A\}\{B\} \rangle$ w.r.t. the maximum length $k=3$ in the C-sequence dataset shown in
\Cref{CoincidenceDB} is $\mathrm{LWU}_3(\langle \{A\}\{B\} \rangle)=20+30+0+0=50$.
\begin{lemma}
\label{lem1}
Given a C-sequence $C$, where $|C| \leq k' \leq k$, then
\begin{equation}
\mathrm{u_{max_k}}(C,k') \leq \mathrm{u_{max_k}}(C,k)
\end{equation} 
\end{lemma}
\begin{proof}
It follows directly from \Cref{MaxUtilK}.  
\end{proof}

\begin{thm}
\label{theorem1}
Given a C-sequence dataset $\delta$ and two L-sequences $L$ and $L'$, where $L \subseteq L'$ and $|L'| \leq k' \leq k$, the following properties hold:
\begin{enumerate}
\item 
\begin{centredequ}
\label{I'}
\mathrm{u_{max}}(L) \leq  \mathrm{LWU}_{|L|}(L)
\end{centredequ}
\item  
	\begin{centredequ}
	\label{I}
	\mathrm{LWU}_k(L') \leq \mathrm{LWU}_k(L)
	\end{centredequ} 
\item 
 \begin{centredequ}
 \label{II}
	\mathrm{LWU}_{k'}(L) \leq \mathrm{LWU}_k(L) 
	\end{centredequ}	
\item 
 \begin{centredequ}
 \label{III}
	\mathrm{LWU}_{k'}(L') \leq \mathrm{LWU}_k(L) 
	\end{centredequ}	
\end{enumerate}

\end{thm}
\begin{proof}
\begin{enumerate}
\item It is inferred from \Cref{MaxUtilEq} and \Cref{LWU_Eq}.
\item The proof of the LDC property in \Cref{I} can be found in \cite{Mirbagheri-DKE}.
\item It trivially follows from \Cref{lem1}.
\item It follows immediately from \Cref{I} and \Cref{II}.
\end{enumerate}
\end{proof}
In order to discover high utility patterns, HUIPMiner generates coincidence candidates by concatenating event labels. As the number of candidates can grow exponentially, the algorithm takes advantage of the LDC property in the pruning strategy, to discard unpromising candidates.   

\begin{defn}
 A coincidence candidate $c$ is \textit{promising} iff $\mathrm{LWU}_k(c) \geq \xi$. Otherwise it is \textit{unpromising}.   
\end{defn}
\begin{corollary} 
\label{corol_1}
Let $\mathrm{a}$ be an unpromising coincidence candidate and $\mathrm{a'}$ be a coincidence. Any superset produced by concatenating $\mathrm{a}$ and $\mathrm{a'}$ is of low utility.   
\end{corollary}
\begin{proof} 
It follows directly from the LDC property.
\end{proof} 

\section{The Projected Utilization}
\label{Section3}
In this section, we introduce a new upper bound called \textit{projected utilization} of an L-sequence, $\mathcal{P}_k$, and we show that $\mathcal{P}_k$ is a tighter upper bound compared to $\mathrm{LWU_k}$.
\begin{defn}
\label{ProjectedDef}
($\mathcal{P}_{k}$) The projected utilization of $L$ w.r.t. a maximum length $k$ is defined as sum of the maximum utility of $L$ with the L-sequence-weighted utilization of $L$ w.r.t the \textit{remaining length} of $k$:
\begin{equation}
\label{DyanmicEq}
\mathcal{P}_{k}(L)=\mathrm{u_{max}}(L)+\mathrm{LWU}_{k-|L|}(L)  
\end{equation}
where $|L| \leq k$ denote the length of L-sequence $L$.
\end{defn}
For example, the projected utilization of $L=\langle \{A\}\{B\} \rangle$ w.r.t. the maximum length $k=3$ in the C-sequence dataset shown in \Cref{CoincidenceDB} is $\mathcal{P}_3(\langle \{A\}\{B\} \rangle)=\mathrm{u_{max}}(\langle \{A\}\{B\} \rangle)+\mathrm{LWU}_{1}(\langle \{A\}\{B\} \rangle)=(13 + 9 + 0 + 0)+(8+12+0+0)=42$.

In contrast to $\mathrm{LWU}_{k}$, which remains constant during the process of discovery for an L-sequence, $\mathcal{P}_{k}$ is dynamically decreasing with respect to the maximum length of the expected patterns. As the length of the pattern gets closer to the maximum length, the maximum utility for the pattern will be projected (decreased), which causes a reduction in the search space for finding the remaining part of the pattern. This will lead us to the following theorem.  
\begin{lemma}
\label{Projected-thm1}
$\mathcal{P}_{k}(L)$ is upper bounded by $\mathrm{LWU}_k (L)$. More formally,
\begin{equation}
\label{Propos}
\mathcal{P}_{k}(L) \leq \mathrm{LWU}_k (L)
\end{equation}
\end{lemma}

\begin{proof} 
We rewrite \Cref{Propos} in accordance with \Cref{ProjectedDef}:
\begin{align*}
& \mathcal{P}_{k}(L)=\mathrm{u_{max}}(L)+\mathrm{LWU}_{k-|L|}(L)  \leq \mathrm{LWU}_k (L)& \nonumber \\
&\Rightarrow \mathrm{u_{max}}(L) \leq \mathrm{LWU}_k (L) - \mathrm{LWU}_{k-|L|}(L)= \mathrm{LWU}_{|L|}(L)& \nonumber \\
\label{Inference}
&\Rightarrow \mathrm{u_{max}}(L) \leq  \mathrm{LWU}_{|L|}(L).&
\end{align*}
\end{proof}

\begin{thm}[utility-Projected Downward Closure property]
\label{Projected-thm}
Given a C-sequence dataset $\delta$ and two L-sequences $L$ and $L'$, where $L \subseteq L'$ and $|L'| \leq k$, then
\begin{equation}
\mathcal{P}_{k}(L') \leq \mathcal{P}_{k}(L)
\end{equation} 
\end{thm}
\begin{proof}
\begin{align*}
& \xRightarrow{\text{\Cref{ProjectedDef}}} \mathrm{u_{max}}(L')+\mathrm{LWU}_{k-|L'|}(L') \leq \mathrm{u_{max}}(L)+\mathrm{LWU}_{k-|L|}(L) & \\
& \xRightarrow{\text{\Cref{I'}}}
\mathrm{LWU}_{|L'|}(L')+ \mathrm{LWU}_{k-|L'|}(L') \leq  \mathrm{LWU}_{|L|}(L) + \mathrm{LWU}_{k-|L|}(L) & \\
&\Rightarrow \mathrm{LWU}_{k}(L') \leq   \mathrm{LWU}_{k}(L).
\end{align*}
\end{proof}

We now redefine the promising and unpromising candidates based on the PDC property (\Cref{Projected-thm}). 
\begin{defn}
\label{newPromising}
 A coincidence candidate $c$ is \textit{promising} iff $\mathcal{P}_{k}(c) \geq \xi$. Otherwise it is \textit{unpromising}.   
\end{defn}
It can be verified that \Cref{newPromising} will not affect \Cref{corol_1} as it now holds by the PDC property. In fact, 
using $\mathcal{P}_{k}$ will lead to fewer or at most the same number of candidates than applying $\mathrm{LWU}_k$.   
The PDC property of $\mathcal{P}_{k}$ will especially be beneficial when finding longer patterns, e.g., patterns of lengths $k \geq 2$,
since as the length of candidates increases, the upper bound $\mathcal{P}_{k}$ keeps reducing.
That makes the search space keeps shrinking which results in a more efficient approach.  

\section{Experiments}
\label{Section4}
We evaluate the effectiveness of the new upper bound, $\mathcal{P}_{k}$, when it is employed by the HUIPMiner algorithm \cite{Mirbagheri-DKE} to mine high utility patterns in interval-based event sequences on a real-world dataset. HUIPMiner was implemented in C++11 and tested 
on a laptop computer with a 2.6GHz Intel 10th generation Core i7 processor and 16GB of memory. 

\subsection{Dataset}
\label{DatasetSec}
We used a publicly available dataset, namely \textit{Blocks} \cite{morchenSensor}, in our experiments. 
 Each event interval in this dataset corresponds to a visual primitive obtained from videos of a human hand stacking colored blocks and describing which blocks are touched as well as the actions of the hand (e.g., contacts blue, attached hand red, etc.). Each e-sequence represents one of eight scenarios, such as assembling a tower.
Since there are no external utilities associated with the dataset, we assume every event label in the dataset has an external utility of 1. \Cref{StatDB} gives a summary of the dataset along with some statistics, including minimum (min), maximum (max), mean (avg), and standard deviation (stdv).

\begin{table*}[ht] 
\centering
\caption{Statistical information about the Blocks dataset }
\label{StatDB}
\begin{tabular}{c|c|c|c|c|c|c|c|c|c }
\hline
  \multicolumn{1}{p{2.75cm}|}{\centering \# \hspace{0pt} Event Intervals  } &  \multicolumn{1}{p{2.10 cm}|}{\centering \hspace{0pt} \#  E-sequences} & \multicolumn{3}{c|}{E-sequence Size} & \multicolumn{1}{p{1.25 cm}|}{\centering \#  Labels} & \multicolumn{4}{c}{Interval  Duration}   \\ 
   		  & &	 min	& max & avg&   &min	& max & avg & stdv   \\ \hline
    1207 & 210 &3 &12 &6 &8 &1&57&17  &12   \\ \hline
 \end{tabular}
\end{table*}

\subsection{Evaluation}
The experiments are conducted to show the effectiveness of the PDC property when it is utilized by a pruning strategy in the HUIPMiner algorithm.  
We evaluate the performance of HUIPMiner when the PDC or LDC properties are used on the Blocks dataset in terms of the execution time and peak memory consumption, while varying the minimum utility threshold $\xi$ and the maximum length of patterns $K$.   
These evaluations are shown on a log-10 scale in \Cref{Pruning_xi} and \Cref{Pruning_k}, respectively. The execution time of HUIPMiner in seconds is shown on the left and the peak memory usage in Kilobytes is presented on the right of the two figures.
The maximum size of patterns $Z$ is set to 5 in the experiments.

\Cref{Pruning_xi} shows the evaluation of the HUIPMiner on the datasets while varying $\xi$ and keeping $K$ set to 4. 
\begin{figure}[!htbp]
\pgfplotsset{width=4.5cm}
\begin{subfigure}{.51 \textwidth}\centering
\begin{tikzpicture}
\begin{axis}[scale only axis,  xtick={0,0.05,...,0.25}, xlabel=$\xi$, ylabel= Time (s), legend pos=north west,  
xticklabel style={ /pgf/number format/fixed, /pgf/number format/precision=5},ymode=log,log basis y={10}]
	\addplot[color=blue,mark=x	] coordinates {
    (0.01,136.637) (0.05,41 ) (0.1,23.3) (0.15,14.933 ) (0.2,7.094) (0.25,5.052)
    };
       \addplot[color=red,mark=o] coordinates {
 	(0.01,122.469) (0.05,36.52 )(0.1,19.419) (0.15, 11.004) (0.2,6.021) (0.25,2.503)						   };
   \end{axis}
\end{tikzpicture}
\label{Blocks_Time_xi}
\end{subfigure}
\begin{subfigure}{.51 \textwidth}\centering
\begin{tikzpicture}
\begin{axis}[scale only axis,legend pos=north west, xtick={0,0.05,...,0.25}, xlabel=$\xi$, xticklabel style={ /pgf/number format/fixed, /pgf/number format/precision=5}, ylabel= Peak Memory Usage (KB)]
	\addplot[color=blue,mark=x	] coordinates {
	(0.01,2536) (0.05, 2112) (0.1,2012) (0.15,1884 ) (0.2,1816) (0.25,1732)     
     };
       \addplot[color=red,mark=o] coordinates {
 	(0.01, 2420) (0.05,1984 ) (0.1,1904) (0.15,1840 ) (0.2,1728) (0.25,1676 )
     };
   \end{axis}
\end{tikzpicture}
\label{Blocks_Memory_xi}
\end{subfigure}


\hspace{10mm}
\begin{tikzpicture} 
      \begin{customlegend}[legend columns=2,legend style={draw=none,column sep=2ex},legend entries={ LDC, PDC}]
            \addlegendimage{color=blue,mark=x}
        \addlegendimage{color=red,mark=o} 
        \end{customlegend}
 \end{tikzpicture}
\caption{Performance Comparison of the HUIPMiner algorithm under various $\xi$}
\label{Pruning_xi}
\end{figure}
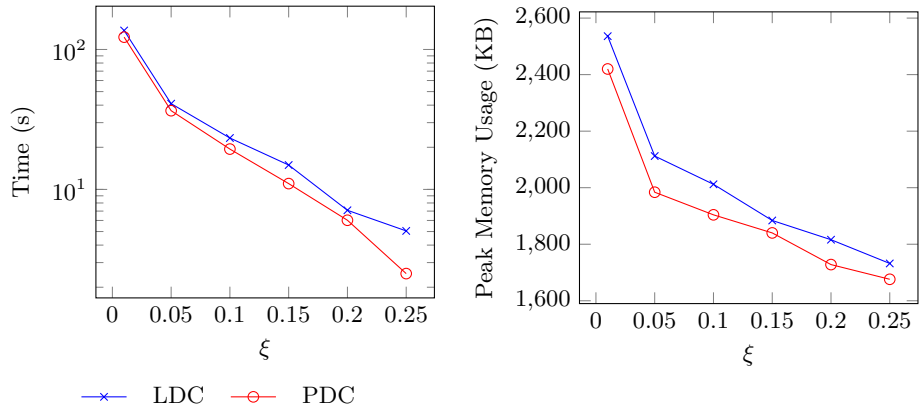
As shown, the execution time of the algorithm is improved by an average of 21\% when the PDC property is used compared to when the LDC is applied. The memory usage is also reduced by an average of 5\% when the PDC property is utilized.  
  
\Cref{Pruning_k} shows the evaluation of the HUIPMiner algorithm on the dataset when $K$ is varied between 1 and 6 and $\xi$ is set to 0.25. 
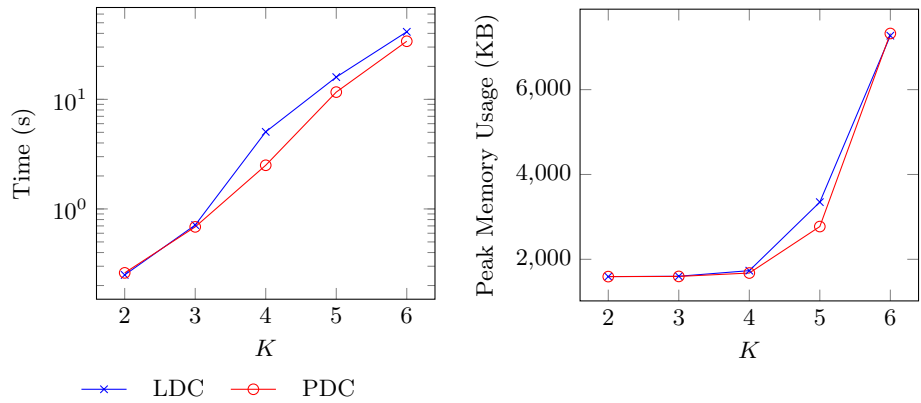
\begin{figure}[!htbp]
\pgfplotsset{width=4.5cm}
\begin{subfigure}{.51 \textwidth}\centering
\begin{tikzpicture}
\begin{axis}[scale only axis,  xtick={2,3,...,6}, xlabel=$K$, ylabel= Time (s), legend pos=north west,  
xticklabel style={ /pgf/number format/fixed, /pgf/number format/precision=5},ymode=log,log basis y={10}]
	\addplot[color=blue,mark=x	] coordinates {
		(2, 0.25) (3,0.713)	(4,5.052)(5, 15.965)(6, 41.335)

    };
       \addplot[color=red,mark=o] coordinates {
	 (2, 0.261) (3,0.686)	(4,2.503)(5,11.645)(6, 33.898)
   				
};
    \end{axis}
\end{tikzpicture}
\label{Blocks_Time_k}
\end{subfigure}
\begin{subfigure}{.51 \textwidth}\centering
\begin{tikzpicture}
\begin{axis}[scale only axis,legend pos=north west,  xtick={2,3,...,6}, xlabel=$K$, xticklabel style={ /pgf/number format/fixed, /pgf/number format/precision=5}, ylabel= Peak Memory Usage (KB)]
	\addplot[color=blue,mark=x	] coordinates {
	(2, 1592) (3,1600)	(4,1732)(5, 3348)(6, 7272)

     };
       \addplot[color=red,mark=o] coordinates {
(2, 1592) (3,1596)	(4,1676)(5,2772)(6, 7324)	
						
 	     };
   \end{axis}
\end{tikzpicture}
\label{Blocks_Memory_k}
\end{subfigure}


\hspace{10mm}
\begin{tikzpicture} 
      \begin{customlegend}[legend columns=2,legend style={draw=none,column sep=2ex},legend entries={ LDC, PDC}]
            \addlegendimage{color=blue,mark=x}
        \addlegendimage{color=red,mark=o} 
        \end{customlegend}
 \end{tikzpicture}
\caption{Performance Comparison of the HUIPMiner algorithm under various $K$}
\label{Pruning_k}
\end{figure}
The results of these experiments indicate that using the PDC property will improve the execution time and memory usage of the algorithm by an average of 19\% and 4\%, respectively. Interestingly, when $K=4$, the algorithm can perform two times faster by utilizing the projected upper bound.   

The number of extracted high utility patterns is also tested to ensure that the projected upper bound does not compromise the completeness of the algorithm. \Cref{figNPat} confirms the integrity of using both upper bounds. As expected, we obtained exactly the same number of patterns from the dataset when any of the two upper bounds are applied. 
\begin{figure}[!htbp]
\pgfplotsset{width=4.9cm}
\begin{subfigure}[t]{0.51\textwidth}
\centering
\begin{tikzpicture}
\begin{axis}[scale only axis, xtick={0.01,0.05,0.1,0.15,0.20,0.25},   xlabel=$\xi$, ylabel= Number of Patterns, legend pos=north west,  
tick label style={ /pgf/number format/fixed, /pgf/number format/precision=5}, ymode=log,log basis y={10}, ybar=0.5*\pgflinewidth, bar width=6pt]
	\addplot coordinates {
	(0.01,15548) (0.05,3020 )  (0.1,1252)  (0.15,508)  (0.2,114) (0.25,12)    
    };
   \end{axis}
\end{tikzpicture}
\label{Blocks_Patterns_xi}
\subcaption{$K=4$, $\xi$ is varied}
\end{subfigure}
\begin{subfigure}[t]{0.51\textwidth}
\centering
\begin{tikzpicture}
\begin{axis}[scale only axis, xtick={2,3,...,6},   xlabel=$K$, ylabel= Number of Patterns, legend pos=north west,  
tick label style={ /pgf/number format/fixed, /pgf/number format/precision=5}, ybar=0.5*\pgflinewidth, bar width=6pt]
	\addplot coordinates {
 	(2,6) (3,9 )  (4,12)  (5,23)  (6,61)     
    };
   \end{axis}
\end{tikzpicture}
\label{Blocks_Patterns_k}
\subcaption{$\xi=0.25$, $K$ is varied}
\end{subfigure}
\caption{Number of patterns discovered by the HUIPMiner algorithm when applying the pruning strategy based on either LDC or PDC property}
\label{figNPat}
\end{figure}
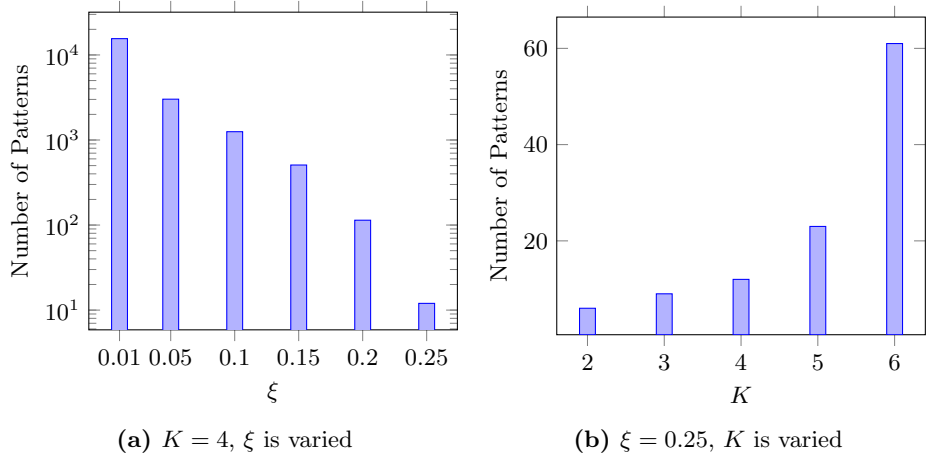

\section{Conclusions}
\label{Section5}
We showed that the projected upper bound can improve the efficiency of the HUIPMiner algorithm. 
By applying the projected upper bound, HUIPMiner can be executed up to two times faster than when LWU is applied. In addition, memory consumption is reduced when the projected upper bound is used.

 \bibliographystyle{splncs} 
\let\chapter\section\bibliography{bibfile} 

\end{document}